\documentclass{biometrika}

\usepackage{amsmath}

\usepackage{times}
\usepackage{bm}
\usepackage{natbib}

\usepackage[plain,noend]{algorithm2e}

\makeatletter
\renewcommand{\algocf@captiontext}[2]{#1\algocf@typo. \AlCapFnt{}#2} 
\renewcommand{\AlTitleFnt}[1]{#1\unskip}
\def\@algocf@capt@plain{top}
\renewcommand{\algocf@makecaption}[2]{%
  \addtolength{\hsize}{\algomargin}%
  \sbox\@tempboxa{\algocf@captiontext{#1}{#2}}%
  \ifdim\wd\@tempboxa >\hsize
    \hskip .5\algomargin%
    \parbox[t]{\hsize}{\algocf@captiontext{#1}{#2}}
  \else%
    \global\@minipagefalse%
    \hbox to\hsize{\box\@tempboxa}
  \fi%
  \addtolength{\hsize}{-\algomargin}%
}
\makeatother

\def\T{{ \mathrm{\scriptscriptstyle T} }}
\def\v{{\varepsilon}}

\begin{document}

\jname{Biometrika}
\accessdate{}
\copyrightinfo{\Copyright\ 2014 Biometrika Trust\goodbreak {\em Printed in Great Britain}}


\markboth{T. Kunihama, A. H. Herring, C. T. Halpern \and D. B. Dunson}{Biometrika style}

\title{Nonparametric Bayes modeling with sample survey weights}

\author{T. Kunihama}
\affil{Department of Statistical Science, Duke University, Durham, North Carolina 27708, U.S.A. \email{tsuyoshi.kunihama@duke.edu}}
\author{A. H. Herring}
\affil{Department of Biostatistics and Carolina Population Center, The University of North Carolina at Chapel Hill, Chapel Hill, North Carolina 27599, U.S.A. \email{aherring@bios.unc.edu}}
\author{C. T. Halpern}
\affil{Department of Maternal and Child Health and Carolina Population Center, The University of North Carolina at Chapel Hill, Chapel Hill, North Carolina 27599, U.S.A. \email{carolyn\_halpern@unc.edu}}
\author{\and D. B. Dunson}
\affil{Department of Statistical Science, Duke University, Durham, North Carolina 27708, U.S.A.\email{dunson@duke.edu}}

\maketitle

\begin{abstract}
In population studies, it is standard to sample data via designs in which the population is divided into strata, with the different strata assigned different probabilities of inclusion.  Although there have been some proposals for including sample survey weights into Bayesian analyses, existing methods require complex models or ignore the stratified design underlying the survey weights.  We propose a simple approach based on modeling the distribution of the selected sample as a mixture, with the mixture weights appropriately adjusted, while accounting for uncertainty in the adjustment.  We focus for simplicity on Dirichlet process mixtures but the proposed approach can be applied more broadly.  We sketch a simple Markov chain Monte Carlo algorithm for computation, and assess the approach via simulations and an application. 
\end{abstract}


\section{Introduction}

In sample surveys, it is routine to conduct stratified sampling designs to ensure that a broad variety of groups are adequately represented in the sample.  In particular, the population is divided into mutually exclusive strata having different probabilities of inclusion.  Analyzing data from such designs is challenging, since the collected sample is not representative of the overall population.  To correct for discrepancies in the statistical analysis, survey weights are constructed.  However, it is unclear how to appropriately include these weights, particularly in Bayesian analyses.

\cite{Little04} and \cite{Gelman07} clarify the importance of including survey weights into model-based analyses. \cite{ZhengLittle03, ZhengLittle05} propose a nonparametric spline model and  \cite{ChenElliottLittle10} extend the framework for binary variables.  An unpublished 2014 technical report by Y. Si, N. Pillai and A. Gelman propose a nonparametric model in which the survey weights are linked with a response through a Gaussian process regression.  These approaches can flexibly connect the survey weights with  the response.  However, they require additional modeling of survey weights for non-sampled subjects in the population, leading to highly complex models.

In this article, we propose a simple approach in which we apply standard mixture models for the selected sample, and then adjust the mixture weights based on the survey weights.  We allow probabilistic uncertainty in this adjustment in a Bayesian manner.  Posterior computation relies on a simple modification to add an additional step to Markov chain Monte Carlo algorithms for mixture models.

\section{Mixture Models with Survey Weights}
\subsection{Adjusted density estimates} \label{sec:kernel}

Let $y_1,\ldots,y_N$ denote independently and identically distributed observations from a density $f_0$ with 
$y_i\in \mathcal{R}$ for $i \in D=\{1,\ldots,N\}$.  From this initial population, $n$ subjects are sampled, with $w_i = c/\pi_i$ the survey weight for subject $i$, $c$ a positive constant, and $\pi_i$ the inclusion probability for $i \in D$.  We assume $D$ can be divided into mutually exclusive subpopulations $D_1,\ldots,D_M$, with $\{y_i, i\in D_m\}$ independently and identically distributed from density $f_m$, for $m=1,\ldots,M$.  Then, $f_0$ can be expressed as 
\begin{align}
f_0(y) = \sum_{m=1}^{M} \nu_m f_m(y), 
\label{eq:true}
\end{align}
where $\nu_m\geq 0$ and $\sum_{m=1}^{M} \nu_m = 1$.  By applying kernel density estimation to each $f_m$ in (\ref{eq:true}), \cite{Buskirk98} and \cite{BellhouseStafford99} propose an adjusted density estimate,
\begin{align}
\hat{f}_0(y) = \sum_{i \in S} \frac{\tilde{w}_i}{b} \mathcal{K}\left( \frac{y-y_i}{b} \right), 
\label{eq:kernel}
\end{align}
where $S \subset D$ are the selected subjects in the survey, $\tilde{w}_i=w_i/\sum_{j\in S} w_j$, $\mathcal{K}$ is a kernel function and $b>0$.  Estimator (\ref{eq:kernel}) adjusts for bias in the usual kernel estimator applied to sample $S$ by modifying the weight for the $i$th subject from $1/n$ to $\tilde{w}_i$.  This adjustment leads to consistency under some conditions (\cite{BuskirkLohr05}).

\subsection{Bayesian adjustments with uncertainty}

Section 2$\cdot$1 focuses on univariate continuous variables, while our goal is to develop a general approach for adjusting posterior distributions to take into account sample survey weights.  Let $y \in \mathcal{Y}$ denote a random variable, with $\mathcal{Y}$ a Polish space that may correspond to a $p$-dimensional Euclidean space, a discrete space, a mixed continuous and discrete space, a non-Euclidean Riemannian manifold, such as a sphere, and other cases.  Extending (\ref{eq:true}) to general spaces, we let $f_0(\cdot)$ and $f_m(\cdot)$, for $m=1,\ldots,M$, denote densities on $\mathcal{Y}$ with respect to a dominating measure $\mu$.  The density in the $m$th subpopulation is expressed as a mixture, 
\begin{align}
f_m(y) = \sum_{h=1}^H \nu_{mh} f(y\mid \theta_h),
\label{eq:fm}
\end{align}
where $\nu_{mh}\geq 0$, $\sum_{h=1}^H\nu_{mh}=1$ and $\theta_h$ are parameters characterizing the $h$th mixture component.  Then, $f_0$ can be approximately expressed as a mixture having the same kernels as in (\ref{eq:fm}) but with adjusted weights as in (\ref{eq:kernel}).

\begin{theorem}
Let $s_i \in\{1,\ldots, H\}$ denote the mixture index for subject $i$ for $i \in S$. Let $S_h = \{ i: s_i = h, i \in S\}$, for $h=1,\ldots,H$.  Then, for large $N$ and $n$, 
\begin{align}
f_0(y) &\approx \sum_{h=1}^H \frac{ \sum_{i\in S_h} w_i /c }{N} f(y\mid \theta_{h}) \approx \sum_{i\in S} \tilde{w}_i f(y \mid \theta_{s_i}).
\label{eq:true-3} 
\end{align}
\end{theorem}

\begin{proof}
Letting $N_m$ be the number of subjects in $D_m$, $N_m/N \rightarrow \nu_m$ by the law of large numbers. Letting $w^*_m$ and $\pi^*_m$ denote the survey weight and inclusion probability for the $m$th subpopulation, $w_i=w^*_m$ and $\pi_i=\pi^*_m$ for $i\in D_m$.  From (\ref{eq:true}) and (\ref{eq:fm}), $f_0$ can be expressed as
\begin{align}
f_0(y) &= \sum_{m=1}^M \nu_m f_m(y) \approx \sum_{m=1}^M \frac{N_m}{N} f_m(y)
= \sum_{h=1}^H \sum_{m=1}^M \frac{N_m \nu_{mh}}{N} f(y\mid \theta_h) \nonumber \\
&\approx \sum_{h=1}^H \frac{\sum_{i\in S_h} w_i /c }{N} f(y\mid \theta_h) 
\approx \sum_{i\in S} \tilde{w}_i f(y \mid \theta_{s_i}).
\label{eq:true-p}
\end{align}
The first approximation in (\ref{eq:true-p}) can be induced by $N_m\approx w_m^*n_m/c$ and
\begin{align*}
\nu_{mh} \approx \frac{\sum_{i\in S}1(i\in D_{mh})}{n_m}, 
\end{align*}
for large $N_m$ and $n_m$, where $D_{mh}$ is a subset of $D_m$ with $s_i=h$.  The second approximation in (\ref{eq:true-p}) is based on $c\approx \sum_{i\in S}w_i/N$, which is derived by summation of $N_m\approx w_m^*n_m/c$ over $m$. 
\end{proof}
 
\noindent
Under random designs with $w_i\propto c$, $f_0$ can be approximated by
\begin{align}
f_R(y) = \sum_{h=1}^H \frac{\sum_{i\in D} 1(i\in S_h) }{n} f(y\,|\,\theta_h) = \sum_{i\in S} \frac{1}{n} f(y\mid \theta_{s_i}). 
\label{eq:random}
\end{align}
Comparing the last terms in (\ref{eq:true-3}) and (\ref{eq:random}), we can interpret that the bias can be adjusted by 
shifting the weight for the $i$th sampled subject from $1/n$ to $\tilde{w}_i$ as in  (\ref{eq:kernel}).

We propose a simple Bayesian adjustment method using the second term in (\ref{eq:true-3}).  We consider a standard Bayesian mixture model,
\begin{align}
f_B(y) &= \sum_{h=1}^H \lambda_h f(y\mid \theta_h), \ \ \lambda \sim \pi(\lambda), \ \ \theta_h \sim \pi(\theta_h), 
\label{eq:miture1}
\end{align}
where $\lambda=(\lambda_1,\ldots,\lambda_H)^{\T}$ with $\lambda_h \geq 0$ and $\sum_{h=1}^H\lambda_h=1$, and $\pi(\lambda)$ and $\pi(\theta_h)$ are priors for $\lambda$ and $\theta_h$.  For example, using a truncated stick-breaking process (\cite{IshwaranJames01}), we let $\lambda_h = V_h \prod_{l<h} (1-V_l)$, $V_h \sim \mbox{Be}(1,\alpha)$ for $h=1,\ldots,H-1$ with $V_H=1$.  However, our focus is not on the specific mixture model and prior but on the adjustment for sampling bias, and alternative priors can be used without complication.

Comparing the second terms in (\ref{eq:true-3}) and (\ref{eq:random}), the difference is in the mixture weights.  The expression in (\ref{eq:true-3}) can be interpreted as implying that $\sum_{i\in S_h} w_i/c$ subjects are generated from the $h$th mixture component in the population.  Updating the prior 
$\tilde{\lambda} \sim \mbox{Dir}(a_1,\ldots,a_H)$ with this information, we obtain the following conditional posterior distribution for the adjusted weights $\tilde{\lambda}=(\tilde{\lambda}_1,\ldots,\tilde{\lambda}_H)^{\T}$,
\begin{align}
\tilde{\lambda} \sim \text{Dir}\left( a_1 + \frac{1}{\tilde{c}} \sum_{i:s_i=1} w_i, \ldots, a_H + \frac{1}{\tilde{c}} \sum_{i:s_i=H} w_i \right),
\label{eq:step}
\end{align}
where $\tilde{c} = \sum_{i\in S} w_i / N \approx c$.  Expression (\ref{eq:step}) takes into account uncertainty in the adjusted weights in mixture component allocation.  Even as the population size $N$ becomes large, there may be certain mixture components that are not represented in the selected sample, leading to substantial uncertainty in the adjustment.  Posterior computation is straightforward: we simply apply any existing Markov chain Monte Carlo algorithm for mixture models to the selected sample, add sampling step (\ref{eq:step}) for generating the adjusted weights $\tilde{\lambda}$, and apply this adjustment to each step of the sampling algorithm to obtain samples from an adjusted posterior for the population density $f_0(y)$.  
As a default, we set $a_h = a$ for $h=1,\ldots,H$, with prior sample size $Ha \sim 1-2\%$ of population size $N$.

\section{Simulation Study}

We illustrate performance of the proposed approach and compare to competitors.  We consider three cases in which a population with $N=1,000,000$ consists of three subpopulations having $N_1=650,000$, $N_2=300,000$ and $N_3=50,000$ with $\nu_m=N_m/N$.  From each stratum, we randomly generate $n_m=500$ subjects and construct survey weights by $w_i=N_m/n_m$ for $i\in D_m$ for $m=1,2,3$.  As competitors, we employ three model-based Bayesian methods.  First, we consider a model-based Horvitz-Thompson estimator (\cite{HorvitzThompson52}; \cite{Little04}), $y_i = \beta \pi_i + \varepsilon_i$, $\varepsilon_i \sim N(0, \pi_i^2 \sigma^2)$ where $\pi_i = 1/w_i$.  Second, we consider a polynomial regression with random effects, $y_i = \beta_0 + \beta_1 \pi_i + \beta_2 \pi_i^2 + \gamma_{[i]} + \varepsilon_i$, $\varepsilon_i \sim N(0, \pi_i^2 \sigma^2)$, $\gamma_m \sim N(0, \tau^2)$ where $\gamma_{[i]}$ denotes a random effect for the subpopulation to which the $i$th subject belongs.  This can be induced by the spline model of  \cite{ZhengLittle03}.  Also, we apply the Gaussian process regression model from a 2014 technical report by Y. Si, N. Pillai and A. Gelman, $y_i = \mu(x_{[i]}) + \varepsilon_i$, $\varepsilon_i \sim N(0, \sigma^2)$, $\mu(x) \sim \text{GP}( \beta x, C)$, $C\{ \mu(x_m), \mu(x_{m'}) \} = \text{cov}\{\mu(x_m), \mu(x_{m'})\} = \tau^2 \exp(-\kappa |x_m - x_{m'}|)$ where $x_m=\log(w^*_m)$ and $x_{[i]}$ denotes the log weight for the
stratum for the $i$th subject.  We also apply Dirichlet process mixtures without weight adjustment.

\begin{figure}[t]
\figurebox{21pc}{35pc}{}[simulation-2.eps]
\caption{Estimated densities in case 1. Green lines with squares are the true density, red lines with circles the posterior means and red dash lines 95\% credible intervals.  Proposed means the proposed method, Non-adjusted the Dirichlet process mixtures without weight adjustment, HT Horvitz-Thompson estimator, RE polynomial regression with random effects and GP Gaussian process regression.}
\label{fig:density-1}
\end{figure}

\begin{figure}[htbp]
\figurebox{21pc}{35pc}{}[simulation-1.eps]
\caption{Estimated densities in case 2. Green lines with squares are the true density, red lines with circles the posterior means and red dash lines 95\% credible intervals.  Proposed means the proposed method, Non-adjusted the Dirichlet process mixtures without weight adjustment, HT Horvitz-Thompson estimator, RE polynomial regression with random effects and GP Gaussian process regression.}
\label{fig:density-2}
\end{figure}

\begin{figure}[htbp]
\figurebox{21pc}{35pc}{}[simulation-3.eps]
\caption{Estimated probabilities in case 3. Green lines with squares are the true density, red lines with circles the posterior means and red dash lines 95\% credible intervals.  Proposed means the proposed method, Non-adjusted the Dirichlet process mixtures without weight adjustment, HT Horvitz-Thompson estimator, RE polynomial regression with random effects and GP Gaussian process regression.}
\label{fig:density-3}
\end{figure}

In the first case, we assume $f_1(y)= f_N(y\,|2,0.6)$, $f_2(y)=f_N(y\,|0,0.4)$ and  $f_3(y)=f_N(y\,|-2,0.3)$ in (\ref{eq:true}) where $f_N(y\,|\,a,b)$ denotes a normal density with mean $a$ and standard deviation $b$.  For the proposed method, we use the Dirichlet process mixture of normals, $f_B(y)=\sum_{h=1}^H \lambda_h f_N(y\mid \mu_h, \tau_h)$ where $\lambda_h=V_h(1-V_l)$, $V_h\sim \text{Be}(1,\alpha)$, $V_H=1$ with $H = 20$, $\alpha \sim \text{Ga}(0.25,0.25)$, $\mu_h \sim N(\bar{y},s_y^2)$, $\tau_h^2\sim \text{Inverse-Gamma}(2,s_y^2/2)$ where $\bar{y}$ and $s^2_y$ are the sample mean and variance.  As for the prior in the step (\ref{eq:step}), we set $a_h=1,000$ for each $h$.  For competitors, we assume the following priors: $\beta\sim N(0,s_y^2)$, $\beta_j \sim N(0,s_y^2)$, $\sigma^2\sim \text{Inverse-Gamma}(2,s_y^2/2)$, $\tau^2\sim \text{Inverse-Gamma}(2,s_y^2/2)$ and $\kappa\sim \text{Ga}(1,2)$.  We draw 10,000 samples after the initial 5,000 samples are discarded as a burn-in period and every 10th sample is saved.  Rates of convergence and mixing were adequate.  Figure \ref{fig:density-1} shows the estimation results for case 1.  The Horvitz-Thompson estimator fails to capture the multimodality, while the non-adjusted estimator has considerable bias.  The random effect model and Gaussian process have somewhat better performance, but clear bias remains.  The proposed method accurately estimates the density, and 98\% of true values are covered in the 95\% credible intervals across 100 equally spaced grid points in [-6, 6].

We also considered a more complex density for each stratum, $f_1(y)=0.2 f_N(y\,|-2,1)+0.8f_N(y\,|\,2,0.8)$, $f_2(y)=0.4 f_N(y\,|-2,1)+0.6f_N(y\,|\,2,0.8)$ and  $f_3(y)=0.85 f_N(y\,|-2,1)+0.15f_N(y\,|\,2,0.8)$.  The Markov chain Monte Carlo settings are the same as in case 1.  Figure \ref{fig:density-2} reports the result for case 2.  The Horvitz-Thompson estimator, random effect model and Gaussian process regression work poorly, missing the multimodal shape of the true density because they construct population densities relying on unimodal densities for subpopulations.  The non-adjusted method capture the bimodality but with substantial bias.  The proposed method approximates the density well, while covering 100\% of true values in the 95\% intervals.

We also consider a mixture of Poisson distributions, $f_1(y)=0.2 \text{Poisson}(y\,|15)+0.8\text{Poisson}(y\,|4)$, $f_2(y)=0.4 \text{Poisson}(y\,|15)+0.6 \text{Poisson}(y\,|4)$ and  $f_3(y)=0.85 \text{Poisson}(y\,|15)+0.15 \text{Poisson}(y\,|4)$.  For the Dirichlet process mixtures, we apply the rounded kernel method in \cite{CanaleDunson11} where latent continuous variables are modeled by (\ref{eq:miture1}) with the same Markov chain Monte Carlo settings as in case 1.  Also, we apply the competitors to log transformed observations $y^*_i=\log(y_i+0.5)$ and estimate probabilities by $\text{pr}(y_i=y)=\text{pr}\{\log(y)< y^*_i\leq \log(y+1)\}$ for $y=0,1,\ldots,\infty$.  Figure \ref{fig:density-3} shows the result. We observe the proposed method obtains good approximation, while the competitors fail to capture the mode at 15.  Also, 98\% of the true values are covered in the 95\% intervals in the support from 0 to 100.

To assess the impact of increasing the number of strata while decreasing within-strata sample size, we consider a case with $M=100$ in which $N_m=1000m$, $n_m=20$ for $m=1,\ldots,100$ with $N=5,050,000$ and $n=2,000$ and $f_m(y)=f_N(y\,|-2,0.3)$ for $m=1,\ldots,30$, $f_m(y)=f_N(y\,|0,0.4)$ for $m=31,\ldots,70$ and $f_m(y)= f_N(y\,|2,0.6)$ for $m=71,\ldots,100$.  We obtain a similar result to case 1 with the proposed method dominating competitors.

\section{Application to Adolescent Behaviour Analysis}

\begin{figure}[t]
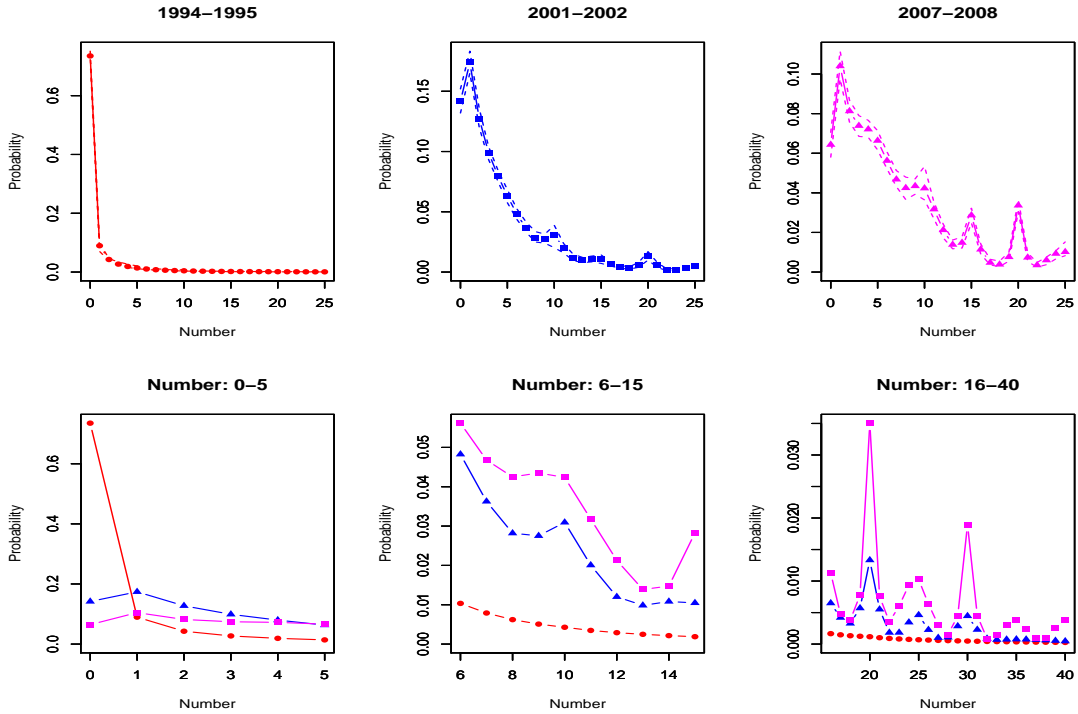

\figurebox{23.5pc}{35pc}{}[Short-real.eps]
\caption{Estimated probabilities of total numbers of sex partners. The first row shows estimated probabilities in 1994-1995 (left), 2001-2002 (middle) and 2007-2008 (right). Lines with symbols show posterior means and dash lines 95\% credible intervals.  The second row shows posterior means of 0-5 partners (left), 6-15 (middle) and 15-40 (left).  Red lines with circle represent posterior means for 1994-1995, blue lines with triangles for 2001-2002 and purple lines with squares for 2007-2008.}
\label{fig:density}
\end{figure}

We apply the proposed method to the National Longitudinal Study of Adolescent Health.  Our focus is on studying the total number of sex partners in adolescence.  The target population is adolescents in grades 7-12 in the United States during the 1994-95 school year with $N=14,677,347$.  The full study design is described by \cite{Harris09}.  The study drew supplemental samples, oversampling groups of particular interest based on ethnicity, genetic relatedness to siblings, adoption status, disability, and black adolescents with highly educated parents.  We use three waves of surveys in which participants are in grades 7-12 (1994-1995), young adults age 18-26 (2001-2002) and adults age 24-32 (2007-2008).  In each wave, numbers of observations are 6447, 4812 and 4819, respectively.  We use the rounded kernel method with Dirichlet process mixtures as in the simulation.  Since we expect high right skew in these data, we use log cut-points instead of non-negative integers, so that the Dirichlet process mixtures can efficiently approximate such distributions.  For the priors of the latent continuous variable, we use $\mu_h \sim N(\tilde{y},\tilde{s}_y^2)$, $\tau_h^2\sim \text{Inverse-Gamma}(2,\tilde{s}_y^2/200)$ where $\tilde{y}$ and $\tilde{s}^2_y$ are the sample mean and variance of $\log(y_i+0.5)$.  Also, we set $a_h=10,000$ for the Dirichlet prior in (\ref{eq:step}).  We draw 20,000 samples after the initial 5,000 samples are discarded as a burn-in period and every 10th sample is saved.  We observe that the sample paths were stable and the sample autocorrelations dropped smoothly.

Figure \ref{fig:density} shows the estimated probabilities for the three waves.  1994-1995 shows a high probability on zero with small values for positive counts.  2001-2002 expresses differences from 1994-1995 in that the probability on zero considerably decreases, while one shows the highest value and the tail gets heavy.  The shape in 2007-2008 is similar to 2001-2002 in that both have highest probabilities at one and then steep declines.  2007-2008 shows a heavier tail with relatively high spikes at multiples of five.  This is probably because people with many partners do not remember the exact numbers.



\section*{Acknowledgement}
This work was supported by Nakajima Foundation and grants from the National Institutes of Health.  The data are from Add Health, a program project directed by Kathleen Mullan Harris and designed by J. Richard Udry, Peter S. Bearman, and Kathleen Mullan Harris at the University of North Carolina at Chapel Hill, and funded by grants from the Eunice Kennedy Shriver National Institute of Child Health and Human Development, with cooperative funding from 23 other federal agencies and foundations. Special acknowledgment is due Ronald R. Rindfuss and Barbara Entwisle for assistance in the original design. The data are available at http://www.cpc.unc.edu/addhealth. 

\bibliographystyle{biometrika}
\bibliography{ah}

\end{document}